\documentclass[letterpaper, 10 pt, conference]{cssconf}
\IEEEoverridecommandlockouts 
\usepackage{ellipsis, mparhack, ragged2e} 
\usepackage[l2tabu, orthodox]{nag} 

\usepackage[english]{babel}
\usepackage[T1]{fontenc}
\usepackage{lmodern}
\usepackage[utf8]{inputenc}
\usepackage[final]{microtype}

\usepackage{amsmath, amssymb}
\usepackage[hyperref, amsmath]{ntheorem}
\newtheorem{theorem}{Theorem}
\newtheorem{definition}{Definition}
\theoremstyle{plain}
\newtheorem{remark}{Remark}
\newtheorem{cor}{Corollary}

\usepackage[caption=false]{subfig}
\usepackage[autostyle=true]{csquotes}

\PassOptionsToPackage{hyphens}{url}\usepackage[hidelinks]{hyperref}

\usepackage{tikz}
\usepackage{booktabs}
\usepackage{tabu,multirow}
\usepackage{enumerate}
\usepackage{rotating}
\usepackage{adjustbox}
\newcolumntype{L}{X}
\newcolumntype{R}{>{\raggedleft\arraybackslash}X}
\newcolumntype{C}{>{\centering\arraybackslash}X}

\usepackage[capitalize]{cleveref}
\newcommand{\todo}[1]{\fxnote{#1}}
\renewcommand{\todo}[1]{}

\usepackage{algorithmicx,algorithm}
\usepackage[noend]{algpseudocode}

\DeclareFontFamily{U}{mathx}{\hyphenchar\font45}
\DeclareFontShape{U}{mathx}{m}{n}{
	<5> <6> <7> <8> <9> <10>
	<10.95> <12> <14.4> <17.28> <20.74> <24.88>
	mathx10
}{}
\DeclareSymbolFont{mathx}{U}{mathx}{m}{n}
\DeclareFontSubstitution{U}{mathx}{m}{n}
\DeclareMathAccent{\widecheck}{0}{mathx}{"71}

\title{\textbf{Satisfiability Bounds for $\omega$-Regular Properties in Bounded-Parameter Markov Decision Processes}}

\author{Maximilian Weininger$^\ast$ \hspace{1em} Tobias Meggendorfer$^\ast$ \hspace{1em} Jan K{\v{r}}et\'insk\'y$^\ast$ \\
Department of Informatics, Technical University of Munich%
\thanks{$^\ast$All authors contributed equally to this work. It was funded in part by the German Research Foundation (DFG) project KR 4890/2-1 ``Statistical Unbounded Verification''.}%
\vspace{-1em}
}

\input{0_macros}

\colorlet{disabled}{lightgray}
\tikzstyle{state}=[draw,rectangle,inner sep=5pt,rounded corners=2pt]
\tikzstyle{minstate}=[draw,rounded rectangle,inner sep=5pt]
\tikzstyle{maxstate}=[draw,rectangle,inner sep=5pt]

\tikzstyle{action}=[font=\normalsize,inner sep=0pt,outer sep=3pt]
\tikzstyle{actionnode}=[circle,draw=black,fill=black,minimum size=1mm,inner sep=0,outer sep=0]
\tikzstyle{actionedge}=[draw,-]
\tikzstyle{actiontext}=[font=\small,inner sep=0pt,outer sep=2pt]
\tikzstyle{prob}=[font=\normalsize,inner sep=0pt,outer sep=1pt]
\tikzstyle{probedge}=[draw,->]
\tikzstyle{probtext}=[font=\footnotesize,outer sep=2pt,inner sep=0pt]
\tikzstyle{directedge}=[draw,->]
\tikzset{chainarrow/.tip={Stealth[length=3pt]}}
\tikzset{>=chainarrow}

\tikzstyle{automatonstate}=[draw,chamfered rectangle]

\begin{document}

\maketitle
\thispagestyle{empty}
\pagestyle{empty}

\begin{abstract}
We consider the problem of computing minimum and maximum probabilities of satisfying an $\omega$-regular property in a bounded-parameter Markov decision process (BMDP).
BMDP arise from Markov decision processes (MDP) by allowing for uncertainty on the transition probabilities in the form of intervals where the actual probabilities are unknown.
$\omega$-regular languages form a large class of properties, expressible as, e.g., Rabin or parity automata, encompassing rich specifications such as linear temporal logic.
In a BMDP the probability to satisfy the property depends on the unknown transitions probabilities as well as on the policy.
In this paper, we compute the extreme values.
This solves the problem specifically suggested by Dutreix and Coogan in CDC 2018, extending their results on interval Markov chains with no adversary.
The main idea is to reinterpret their work as analysis of interval MDP and accordingly the BMDP problem as analysis of an $\omega$-regular stochastic game, where a solution is provided.
This method extends smoothly further to bounded-parameter stochastic games.
\end{abstract}
\section{Introduction} \label{sec:introduction}

Markov decision processes (MDP) are a classical formalism encompassing both probabilistic and non-deterministic features: in each state some actions are enabled and each action is assigned a distribution over the successor states.
In other words, each action corresponds to a set of transitions, each of which is assigned a transition probability.
\emph{Bounded-parameter MDP (BMDP)} \cite{GLD00} are like MDP, but to each transition is instead assigned an interval of possible transition probabilities.
Thus each BMDP specifies a set of MDP.
There are two interpretations of these intervals.
Firstly, in the \emph{uncertain} interpretation, the BMDP specifies MDP with unknown but constant transition probabilities in the intervals.
An MDP is thus derived from the BMDP by picking a value in each interval once for all.
Secondly, in the \emph{adversarial} interpretation, the BMDP specifies a decision process where the transition probabilities may be different numbers (in the intervals) every time we come to a state.
Each interpretation found its use.
The former can model, for instance, various degrees of uncertainty for each action or confidence intervals for the transition probabilities learnt from experience.
In this case there is one true transition probability, however unknown.
The latter can be used as a formalism for abstracting MDP: states with different outgoing transition probabilities can be abstracted into a single state with an interval covering all the values \cite{GLD00}.
In this case, the interval can stand for any of the values whenever visiting the state.
As such BMDP extend interval Markov chains (IMC) \cite{JL91,KU02} by an adversary (or an underspecified/non-deterministic controller).
The uncertain interpretation of IMC then yields uncertain Markov chains (UMC), while the adversarial interpretation of IMC yields interval MDP (IMDP), as distinguished in \cite{SVA06}.

\emph{$\omega$-regular} languages, e.g.~\cite{DBLP:reference/hfl/Thomas97}, form a robust class of rich specifications, which can be represented in various ways, e.g., by formulae of monadic second-order logic or by automata over infinite words.
In the setting of probabilistic systems, it is often advantageous to use deterministic Rabin automata (DRA) or their variations.
In particular, this class encompasses properties expressible in linear temporal logic (LTL) \cite{DBLP:conf/focs/Pnueli77} and there are efficient ways of translating LTL to DRA \cite{R4}.
Control of MDP with LTL specifications is widely studied, e.g.~\cite{TW86,KB08,WTM10,RDMM14,YTCBB12}, and typically uses the DRA representation.

In \cite{DC18}, Dutreix and Coogan argue for computing minimum and maximum probabilities of satisfying an $\omega$-regular property in an IMC interpreted as IMDP. 
In future work, they wish to apply the technique to solve the problem for BMDP, the controllable counterpart of IMC.
In this paper, we re-interpret their technique in a different light and using that perspective give a solution to BMDP, in both the uncertain and the adversarial understanding of the intervals.
We consider both the upper bound (also called \emph{design choice} of values in intervals \cite{DC18}) and the lower bound (
\emph{antagonistic} in \cite{DC18}).
We present the results for controllers that try to maximize the probability to satisfy the $\omega$-regular property; minimization is analogous as $\omega$-regular languages are closed under complement.

The main idea of our approach is to not only view the IMC as an IMDP, but also as an MDP, since an IMDP is a special case of MDP where the adversary chooses the transition probabilities.
We show the standard analysis on the respective MDP coincides with the tailored algorithm of \cite{DC18} applied to the IMC.
Our main contribution is taking this perspective on BMDP, yielding a stochastic game.
Since we can solve the stochastic game with an $\omega$-regular objective we can obtain also the solutions.
Moreover, for the upper bound, the two players play cooperatively and we can solve the problem in polynomial time using adapted MDP techniques.
Finally, we show how the game extension of IMC with two competing agents can be solved analogously to BMDP, this time without the need of introducing an additional agent.

\paragraph*{Further related work}
The general model of MDP with imprecise parameters (MDPIP) was introduced in \cite{SL73}.
BMDP \cite{GLD00} are then a~subclass where the parametrization is limited to independent intervals.
BMDP have been investigated with respect to various objectives, such as stochastic shortest path (minimum expected reward) \cite{B05}, expected total reward \cite{NG05,WK07,HHH+17}, discounted reward \cite{NG05,FDB14}, LTL \cite{WTM12}, PCTL \cite{PLSS13,HHT16}, or mean payoff 
\cite{TB07}.

The special non-controlled case of IMC has also been investigated for various objectives, e.g. PCTL \cite{SVA06,CK15}, LTL \cite{BLW13} (\cite{BDL+17} observes this algorithm may not converge to the optimum) or $\omega$-regular properties \cite{CSH08}.

Recent improvements include importance sampling techniques for IMC \cite{JWS18} or topological policy iteration for BMDP \cite{RBD19}.
IMC and BMDP are used as abstractions of systems in \cite{LAB15}.

\paragraph*{Organization of the paper}
After recalling the used formalisms in \cref{sec:preliminaries}, we state our problem in \cref{sec:state}.
We provide the solution in \cref{sec:contribution} and an illustrative case study (adjusted from \cite{DC18}) in \cref{sec:cs}.
\cref{sec:conclusion} concludes and presents ideas for future work.

\section{Preliminaries} \label{sec:preliminaries}
In this section, we recall basics of probabilistic systems and set up the notation.
As usual, $\Naturals$ refers to the natural numbers (including 0). 
A \emph{probability distribution} on a finite set $X$ is a mapping $\distribution : X \to [0,1]$, such that $\sum_{x\in X} \distribution(x) = 1$.
We use $\Distributions(X)$ to denote the set of all probability distributions on the set $X$.
Given some set $S$, we use $S^\star$ and $S^\omega$ to denote the set of all finite and infinite sequences comprising elements of $S$, respectively.

\begin{figure*}[t]
	\centering
	\subfloat[Example BMDP]{ \label{fig:model_examples:bmdp}
		\begin{tikzpicture}[auto,initial text=]
			\node[state,initial left] at (-0.5,0) (q0) {$q_0, x$};
			\node[actionnode] at (0.5,0) (q0a) {};
	
			\node[state] at (2,-1) (q1) {$q_1, y$};
			\node[actionnode] at (2,-1.7) (q1a) {};
			\node[actionnode] at (2.5,0) (q1b) {};
	
			\node[state] at (2,1) (q2) {$q_2, z$};
			\node[actionnode] at (0.5,1.5) (q2a) {};
	
			\path[actionedge]
				(q0) edge node[actiontext] {$a$} (q0a)
				(q1) edge[out=-60,in=0] node[actiontext] {$b$} (q1a)
				(q1) edge[in=180,out=90] node[actiontext] {$c$} (q1b)
				(q2) edge node[actiontext] {$d$} (q2a)
			;
	
			\path[probedge]
				(q0a) edge node[probtext,sloped,anchor=north,pos=0.5] {$[0.5, 1.0]$} (q1)
				(q0a) edge node[probtext,sloped,anchor=south,pos=0.5] {$[0.0, 1.0]$} (q2)
	
				(q1a) edge[out=-180,in=-120] node[probtext] {$[0.7, 1.0]$} (q1)
				(q1b) edge[out=0,in=30] node[probtext,sloped,anchor=north,pos=0.7] {$[0.25, 1.0]$} (q1)
				(q1b) edge[out=0,in=-30] node[probtext,sloped,anchor=south,pos=0.7] {$[0.5, 1.0]$} (q2)
	
				(q2a) edge[out=160,in=90,sloped] node[probtext,pos=0.5,anchor=south] {$[0.0, 1.0]$} (q0)
				(q2a) edge[out=160,in=100] node[probtext] {$[0.5, 1.0]$} (q2)
			;
		\end{tikzpicture}
	} %
	\subfloat[MDPs consistent with the BMDP]{ \label{fig:model_examples:mdp}
		\begin{tikzpicture}[auto,initial text=]
			\node[state,initial left] at (-0.5,0) (q0) {$q_0$};
			\node[actionnode] at (0.5,0) (q0a) {};
	
			\node[state] at (2,-1) (q1) {$q_1$};
			\node[actionnode] at (2,-1.7) (q1a) {};
			\node[actionnode] at (2.5,0) (q1b) {};
	
			\node[state] at (2,1) (q2) {$q_2$};
			\node[actionnode] at (0.5,1.5) (q2a) {};
	
			\path[actionedge]
				(q0) edge node[actiontext] {$a$} (q0a)
				(q1) edge[out=-60,in=0] node[actiontext] {$b$} (q1a)
				(q1) edge[in=180,out=90] node[actiontext] {$c$} (q1b)
				(q2) edge node[actiontext] {$d$} (q2a)
			;
	
			\path[probedge]
				(q0a) edge node[probtext,sloped,anchor=north] {$1.0$} (q1)
	
				(q1a) edge[out=-180,in=-120] node[probtext] {$1.0$} (q1)
				(q1b) edge[out=0,in=30] node[probtext,sloped,anchor=north,pos=0.6] {$0.5$} (q1)
				(q1b) edge[out=0,in=-30] node[probtext,sloped,anchor=south,pos=0.6] {$0.5$} (q2)
	
				(q2a) edge[out=160,in=90,sloped] node[probtext,anchor=south] {$0.5$} (q0)
				(q2a) edge[out=160,in=100] node[probtext] {$0.5$} (q2)
			;
		\end{tikzpicture} %
		\begin{tikzpicture}[auto,initial text=]
			\node[state,initial left] at (-0.5,0) (q0) {$q_0$};
			\node[actionnode] at (0.5,0) (q0a) {};
	
			\node[state] at (2,-1) (q1) {$q_1$};
			\node[actionnode] at (2,-1.7) (q1a) {};
			\node[actionnode] at (2.5,0) (q1b) {};
	
			\node[state] at (2,1) (q2) {$q_2$};
			\node[actionnode] at (0.5,1.5) (q2a) {};
	
			\path[actionedge]
				(q0) edge node[actiontext] {$a$} (q0a)
				(q1) edge[out=-60,in=0] node[actiontext] {$b$} (q1a)
				(q1) edge[in=180,out=90] node[actiontext] {$c$} (q1b)
				(q2) edge node[actiontext] {$d$} (q2a)
			;
	
			\path[probedge]
				(q0a) edge node[probtext,sloped,anchor=north] {$0.5$} (q1)
				(q0a) edge node[probtext,sloped,anchor=south] {$0.5$} (q2)
	
				(q1a) edge[out=-180,in=-120] node[probtext] {$1.0$} (q1)
				(q1b) edge[out=0,in=30] node[probtext,sloped,anchor=north,pos=0.6] {$0.25$} (q1)
				(q1b) edge[out=0,in=-30] node[probtext,sloped,anchor=south,pos=0.6] {$0.75$} (q2)

				(q2a) edge[out=160,in=100] node[probtext] {$1.0$} (q2)
			;
		\end{tikzpicture}
	} %
	\caption[]{
		Examples of an BMDP and its instantiations.
		The acceptance is $\RabinAcc = \{\rabinpair{q_2}{q_1}\}$.
	} \label{fig:model_examples}
	\vspace{-1.5em}
\end{figure*}

\subsection{Markov Decision Processes}
%
%
\begin{definition} \label{def:mdp}
	A \emph{Markov decision process (MDP)} is a tuple $\MDP = (\States, \initialstate, \Actions, \stateactions, \mdpTransitions, \RabinAcc)$, where
		$\States$ is a finite set of \emph{states},
		$\initialstate \in \States$ is the \emph{initial} state,
		$\Actions$ is a finite set of \emph{actions},
		$\stateactions: \States \to 2^{\Actions} \setminus \{\emptyset\}$ assigns to every state a non-empty set of \emph{available actions},
		$\mdpTransitions: \States \times \Actions \to \Distributions(\States)$ is a \emph{transition function} that for each state $s$ and (available) action $a \in \stateactions(s)$ yields a probability distribution over successor states, and
		$\RabinAcc \in 2^\States \times 2^\States$ is the \emph{Rabin acceptance}.
	Furthermore, for ease of notation we assume w.l.o.g. that actions are unique for each state, i.e.\ $\stateactions(s) \intersection \stateactions(s') = \emptyset$ for $s \neq s'$.\footnote{The usual procedure to achieve this in general is to replace $\Actions$ by $\States \times \Actions$ and to adapt $\stateactions$ and $\mdpTransitions$ appropriately.}
	An element $(\RabinFin_i, \RabinInf_i) \in \RabinAcc$ is called \emph{Rabin pair}.
	We assume w.l.o.g. that $\RabinFin_i \intersection \RabinInf_i = \emptyset$ for all pairs.
\end{definition}
In figures, we denote a Rabin pair $(\RabinFin, \RabinInf)$ by $\rabinpair{\RabinFin}{\,\RabinInf}$.

An MDP with $\cardinality{\stateactions(s)} = 1$ for all $s \in \States$ is called \emph{Markov chain (MC)}.
For ease of notation, we overload functions that map to distributions $f: Y \to \Distributions(X)$ by $f: Y \times X \to [0, 1]$, where $f(y, x) := f(y)(x)$.
For example, instead of $\mdpTransitions(s, a)(s')$ we write $\mdpTransitions(s, a, s')$ for the probability of transitioning from state $s$ to $s'$ using action $a$.

An \emph{infinite path} in an MDP is an infinite sequence $\infinitepath = s_0 a_0 s_1 a_1 \cdots$, such that $a_i \in \stateactions(s_i)$ and $\mdpTransitions(s_i,a_i, s_{i+1}) > 0$ for every $i \in \Naturals$.
We use $\infinitepath_i$ to refer to the $i$-th state $s_i$ in a given path.
A \emph{finite path} 
is a finite prefix of an infinite path.
$\Inf(\infinitepath) \subseteq \States$ denotes the set of all states which are visited infinitely often in the path $\infinitepath$.

A path $\infinitepath$ is \emph{accepted}, denoted $\infinitepath \accepted \RabinAcc$, if an only if there exists a Rabin pair $(\RabinFin_i, \RabinInf_i) \in \RabinAcc$ such that all states in $\RabinFin_i$ are visited \emph{finitely often}, i.e. $\RabinFin_i \intersection \Inf(\infinitepath) = \emptyset$, and at least one state of $\RabinInf_i$ is visited \emph{infinitely often}, i.e. $\RabinInf_i \intersection \Inf(\infinitepath) \neq \emptyset$.
We call such a Rabin pair \emph{accepting for $\infinitepath$}.

\begin{remark}
	Often, system and specification are modelled separately, e.g., by a labelled MDP together with a description of an $\omega$-regular property such as an LTL formula \cite{DBLP:conf/focs/Pnueli77} or an automaton.
	The common approach then is to build the product of the system, the BMDP, and an automaton describing the specification;
	this results in a system as described in \cref{def:mdp}.
	Since our work is largely independent of this construction's details we omit this step for simplicity.
	We highlight that indeed \cite{DC18} only refers to the product throughout the main body of their work.
	Details can be found in \cref{sec:appendix} and, e.g., \cite{BaierBook}.
\end{remark}

A \emph{policy} (also called \emph{controller}, \emph{strategy}) on an MDP is a function $\strategy : (\States \times \Actions)^*\times S \to \Distributions(\Actions)$ which given a finite path $\finitepath = s_0 a_0 s_1 a_1 \dots s_n$ yields a probability distribution $\strategy(\finitepath) \in \Distributions(\stateactions(s_n))$ on the actions to be taken next.
We call a policy \emph{memoryless randomized} (or \emph{stationary}) if it is of the form $\strategy: \States \to \Distributions(\Actions)$, and \emph{memoryless deterministic} (or \emph{positional}) if it is of the form $\strategy: \States \to \Actions$.
Later in the paper, we prove that positional strategies are indeed sufficient for all considered problems.
We denote the set of all policies of an MDP by $\Strategies$. 
By fixing the policy $\strategy$ in an MDP $\MDP$, we naturally obtain a MC and thus a probability measure $\Probability[\MDP]<\strategy>$ over potential runs \cite{puterman}.
Throughout this work, we are interested in finding policies which maximize the probability of accepting runs, i.e. $\sup_{\strategy \in \Strategies} \Probability[\MDP]<\strategy>[\infinitepath \accepted \RabinAcc]$.

An \emph{end component} in an MDP is a pair $(T, A)$ of a set of states $T$ and a set of actions $A$ such that the system \emph{can} remain within the states $T$ indefinitely, using only actions from $A$.
Formally, a pair $(T, A)$, where $\emptyset \neq T \subseteq \States$ and $\emptyset \neq A \subseteq \Union_{s \in T} \stateactions(s)$, is an end component of an MDP $\MDP$ if
	(i)~for all $s \in T, a \in A \intersection \stateactions(s)$ we have $\{s' \mid \mdpTransitions(s, a, s') > 0\} \subseteq T$, and 
	(ii)~for all $s, s' \in T$ there is a finite path $\finitepath = s a_0 \dots a_n s' \in (T \times A)^* \times T$, i.e.\ the path stays inside $T$ and only uses actions in $A$.
An end component $(T, A)$ is a \emph{maximal end component (MEC)} if there is no other end component $(T', A')$ such that $T \subseteq T'$ and $A \subseteq A'$.
The set of MECs of an MDP $\MDP$ is denoted by $\Mecs(\MDP)$ and can be obtained in polynomial time \cite{CY95}.
For further detail, see \cite[Sec.~10.6.3]{BaierBook}.

\subsection{Bounded-parameter Markov Decision Processes}

\begin{definition}
	A \emph{bounded-parameter Markov decision process (BMDP)} is a tuple $\BMDP = (\States, \initialstate, \Actions, \stateactions, \bmdpTransitionsLower, \bmdpTransitionsUpper, \RabinAcc)$, where
		$\States$, $\initialstate$, $\Actions$, $\stateactions$ and $\RabinAcc$ are as in the definition of MDP, and
		$\bmdpTransitionsLower, \bmdpTransitionsUpper: \States \times \Actions \times \States \to [0, 1]$ are \emph{lower and upper bounds} on the transition probability in each state.
	Again, we assume that actions are unique for each state.
	For consistency, we require that $\sum_{s' \in \States} \bmdpTransitionsLower(s, a, s') \leq 1 \leq \sum_{s' \in \States} \bmdpTransitionsUpper(s, a, s')$ for each state $s$ and action $a \in \stateactions(s)$.
\end{definition}
Given a BMDP $\BMDP = (\States, \initialstate, \Actions, \stateactions, \bmdpTransitionsLower, \bmdpTransitionsUpper, \RabinAcc)$, we call a Markov decision process $\MDP = (\States, \initialstate, \Actions, \stateactions, \mdpTransitions, \RabinAcc)$ \emph{consistent with $\BMDP$}, denoted $\MDP \in \BMDP$, if and only if $\MDP$'s transition probabilities satisfy $\BMDP$'s bounds, i.e. $\bmdpTransitionsLower(s, a, s') \leq \mdpTransitions(s, a, s') \leq \bmdpTransitionsUpper(s, a, s')$ for all states $s, s' \in \States$ and actions $a \in \stateactions(s)$.
Note that in general there are uncountably many MDPs consistent with a BMDP $\BMDP$.
A BMDP with $\cardinality{\stateactions(s)} = 1$ for all $s \in \States$ is called \emph{Interval Markov chain (IMC)} (e.g., \cite{CSH08}). 

See \cref{fig:model_examples:bmdp} for an example BMDP and \cref{fig:model_examples:mdp} for two MDPs consistent with this BMDP.

\subsection{Stochastic Games}

For our analysis, we additionally need the concept of \emph{stochastic games}.
These can be understood as an MDP where, instead of a single agent controlling the process, we have two antagonistic players.
Intuitively, the first player's aim is to obtain an accepted path, while the second player aims to stop the first player from doing so.
Each state in the game is \enquote{owned} by one of the two players and the owner of a particular state can decide which action to take in that state.
\begin{definition}
	A \emph{stochastic game (SG)} is a tuple $\SG = (\States, \initialstate, \Actions, \stateactions, \sgTransitions, \RabinAcc, \ownerFunction)$, where
		$(\States, \initialstate, \Actions, \stateactions, \sgTransitions, \RabinAcc)$ is an MDP and
		$\ownerFunction : \States \to \{\systemPlayer, \environmentPlayer\}$ is an \emph{ownership function}, assigning each state to either player $\systemPlayer$ or player $\environmentPlayer$.
		This naturally gives rise to the sets of states $\StatesSystem$ and $\StatesEnvironment$, which are controlled by the respective player.
\end{definition}
The definitions of (in)finite paths directly extend to stochastic games.
Policies are slightly modified, since each player can only make decisions in a part of the game.
Formally, we have two kinds of policies $\strategySystem : (\States \times \Actions)^* \times \StatesSystem \to \Distributions(\Actions)$ and $\strategyEnvironment : (\States \times \Actions)^* \times \StatesEnvironment \to \Distributions(\Actions)$, one for each player.
We denote the set of all policies of the respective players by $\StrategiesSystem$ and $\StrategiesEnvironment$.
Fixing the policy of a single player yields an MDP, denoted $\SG(\strategy_i)$; fixing both players' policies $\strategySystem$ and $\strategyEnvironment$ again yields an MC and measure over the set of runs, denoted $\Probability[\SG]<\strategySystem, \strategyEnvironment>$ \cite{Con92}.

\section{Problem Statement}\label{sec:state}

In this work, we are given a BMDP and want to control it such that the probability of an accepting run is maximized.
This raises two orthogonal questions:

Firstly, the semantics of the interval constraints have to be fixed.
We consider two different popular interpretations, called \emph{uncertain} and \emph{antagonistic}.
In the \emph{uncertain} (or \emph{design-choice}) model, an external environment fixes the transition probabilities once and for all, i.e. for a BMDP $\BMDP$ a particular consistent MDP $\MDP \in \BMDP$ is chosen.
In the \emph{antagonistic} model, the external environment instead is allowed to change the transition probabilities at every step, taking into account the full path so far.
These interpretations have been shown to be yield the same optima for reachability in interval Markov chains \cite{DBLP:journals/ipl/ChenHK13}.
In the following, we show that this also is the case for BMDP with Rabin objectives; hence we do not distinguish the semantics in our formal problem statement.

Secondly, it is not specified whether the aforementioned environment is cooperative or antagonistic.
We consider both of these two extreme cases.
In particular, we are interested in finding the maximal probability of acceptance while assuming that all transition probabilities are chosen (i)~to our liking and (ii)~as bad as possible.

Formally, given a BMDP $\BMDP$ we want to compute
\begin{enumerate}[(i)]
	\item $\probupper(\BMDP) = \sup_{\strategy \in \Strategies} \sup_{\MDP \in \BMDP} \Probability[\MDP]<\strategy>[\infinitepath \accepted \RabinAcc]$, and \label{item:upper_bound_interpretation}
	\item $\problower(\BMDP) = \sup_{\strategy \in \Strategies} \inf_{\MDP \in \BMDP} \Probability[\MDP]<\strategy>[\infinitepath \accepted \RabinAcc]$. \label{item:lower_bound_interpretation}
\end{enumerate}
Further, we are interested in the optimal policy and the best- / worst-case MDP consistent with the given BMDP, if it exists, i.e. the witnesses for the above values.

Case~(\ref{item:upper_bound_interpretation}) can be understood as a \enquote{design challenge}.
We are interested in building our system, i.e. finding an optimal assignment for all transitions, such that we maximize the probability of acceptance.
On the other hand, Case~(\ref{item:lower_bound_interpretation}) can be thought of as uncertainty about the real world or measurement imprecisions.
Here, we rather are interested in optimizing the worst case and want to find a safe strategy, able to reasonably cope with any concrete instantiation of the intervals.

Consider the situation depicted in \cref{fig:model_examples}.
We are interested in finding the upper and lower bounds for the BMDP given in \cref{fig:model_examples:bmdp}.
The upper bound, $1.0$, is exhibited by the left MDP of \cref{fig:model_examples:mdp}.
There, we end up in $q_1$ with probability $1$ and, by always playing action $b$, get an accepting run with probability $1$.
The right MDP shows the lower bound of $0.5$, since we get stuck in state $q_2$ with probability $0.5$.
Observe that if action $d$ in state $q_2$ had a non-zero probability of moving to $q_0$, the resulting run would be accepting with probability $1$, since we would almost surely eventually reach $q_1$.

In the following, we re-interpret existing approaches, and present our unified approach for solving BMDP.

\section{Solution approach} \label{sec:contribution}

In order to explain our approach, we first shed some light on the simpler case of interval Markov chains (IMC), handled in \cite{DC18}, and provide a different perspective on their approach.

In \cite{DC18}, the authors present a specialized algorithm for dealing with IMCs.
In essence, they compute maximal sets of (non)accepting states and then obtain the final value by solving a reachability query for these states.
We now provide a different viewpoint on their algorithm which will help us solve the more general case of BMDPs.

The key observation is the following.
We can view the intervals in an IMC as a player (the external environment) picking the transition probabilities at every step.
This can be understood as an MDP, where in every state the player has an action for each distribution satisfying the interval constraints.
This MDP has uncountably many actions in general, as there are infinitely many possibilities to choose the probabilities.
However, all possible distributions can be constructed as a convex combination of finitely many special cases which are \emph{basic feasible solutions} (BFS) of a linear program \cite{CSH08,DBLP:journals/tcs/HaddadM18} (also known as \emph{corner-point abstraction}).
We can identify each of these special solutions and use them to construct a finite MDP.
This MDP is often called \emph{interval Markov decision process} (IMDP); not to be confused with BMDP.
By model checking the IMDP, using established methods, we obtain the desired result for the IMC.

Indeed, the algorithm presented in \cite[Sec.~IV-C]{DC18} actually can be interpreted as a symbolic adaption of the standard model checking procedure for Rabin objectives on the MDP, namely an adapted MEC decomposition together with a reachability query \cite[Sec.~10.6.4]{BaierBook}.
Similar to the methods presented in \cite{DBLP:journals/tcs/HaddadM18}, their specialized algorithm cleverly avoids explicit computation of the exponentially large IMDP, achieving polynomial runtime.

Before we can extend these ideas to BMDP, we explain some details of the basic feasible solutions, since they are essential to our idea.

\subsection{Basic feasible solutions}\label{sec:bfs}

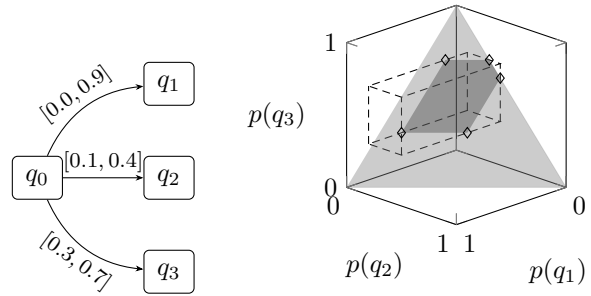
\begin{figure}
	\centering
	\subfloat[Example IMC]{
		\begin{tikzpicture}[auto,initial text=]
		\node[state] at (0,0) (q0) {$q_0$};
		\node[state] at (1.75,1.25) (q1) {$q_1$};
		\node[state] at (1.75,0) (q2) {$q_2$};
		\node[state] at (1.75,-1.25) (q3) {$q_3$};
		
		\path[probedge]
		(q0) edge[sloped, bend left, anchor=south] node[probtext] {$[0.0, 0.9]$} (q1)
		(q0) edge[sloped, anchor=south] node[probtext] {$[0.1, 0.4]$} (q2)
		(q0) edge[sloped, bend right, anchor=north] node[probtext] {$[0.3, 0.7]$} (q3)
		;
		\end{tikzpicture}
	} \quad %
	\subfloat[Visualized constraints.]{
		\begin{tikzpicture}[auto,initial text=]
		\begin{axis}[
		view={135}{20},
		axis lines=box,
		width=4.5cm,height=4.5cm,
		xtick={0,1.0},ytick={0,1.0},ztick={0,1.0},
		xmin=0,xmax=1,ymin=0,ymax=1,zmin=0,zmax=1,
		xlabel={$p(q_1)$},
		ylabel={$p(q_2)$},y label style={anchor=north,rotate=-90},
		zlabel={$p(q_3)$},z label style={rotate=-90}
		]
		\addplot3 [no marks,densely dashed] coordinates {(0.0,0.1,0.3) (0.9,0.1,0.3) (0.9,0.4,0.3) (0.0,0.4,0.3) (0.0,0.1,0.3)};
		\addplot3 [no marks,densely dashed] coordinates {(0.0,0.1,0.7) (0.9,0.1,0.7) (0.9,0.4,0.7) (0.0,0.4,0.7) (0.0,0.1,0.7)};
		
		\addplot3 [no marks,densely dashed] coordinates {(0.0,0.1,0.3) (0.0,0.1,0.7)};
		\addplot3 [no marks,densely dashed] coordinates {(0.9,0.1,0.3) (0.9,0.1,0.7)};
		\addplot3 [no marks,densely dashed] coordinates {(0.0,0.4,0.3) (0.0,0.4,0.7)};
		\addplot3 [no marks,densely dashed] coordinates {(0.9,0.4,0.3) (0.9,0.4,0.7)};
		
		\addplot3 [fill=gray, opacity=0.1, fill opacity=0.4] coordinates {(1.0,0.0,0.0) (0.0,1.0,0.0) (0.0,0.0,1.0)};
		
		\addplot3 [mark=diamond, fill=darkgray, draw=none, opacity=1.0, fill opacity=0.4] coordinates {(0.0,0.3,0.7) (0.2,0.1,0.7) (0.6,0.1,0.3) (0.3,0.4,0.3) (0.0,0.4,0.6)};
		\end{axis}
		\end{tikzpicture}
	}
	
	\caption{Visualization of the set of basic feasible solutions.
		The left picture shows a state together with its transition constraints in an IMC, the right picture depicts a geometric representation of the inequalities induced by the interval constraints.
		The light grey plane corresponds to the distribution constraints (\cref{enum:bfs_constraint:sum}), while the dashed box represents the interval constraints (\cref{enum:bfs_constraint:inequality}).
		Finally, the dark grey area shows the set of consistent distributions and the diamonds mark the basic feasible solutions.
	} \label{fig:bfs}
\end{figure}

Given an IMC $(\States, \initialstate, \bmdpTransitionsLower, \bmdpTransitionsUpper, \labelling)$ and a state $s \in \States$, we are interested in the set of all successor distributions $p \in \Distributions(\States)$ consistent with the constraints imposed by the IMC.
In particular, the following constraints have to be satisfied:
\begin{enumerate}[(I)]
	\item \label{enum:bfs_constraint:sum} $\sum_{s' \in \States} p(s') = 1$, and
	\item \label{enum:bfs_constraint:inequality} $\bmdpTransitionsLower(s, s') \leq p(s') \leq \bmdpTransitionsUpper(s, s')$ for all $s' \in \States$.
\end{enumerate}
Geometrically, \cref{enum:bfs_constraint:sum} constrains the solution set to a plane, while \cref{enum:bfs_constraint:inequality} bounds it in a box, as shown in \cref{fig:bfs}.
Since each point in the solution corresponds to a valid transition distribution and vice versa, we identify solutions of the constraints with their corresponding distributions.
Observe that the resulting solution set (and the corresponding set of distributions) is convex and any element of this set can be obtained by a convex combination of the corner-points, which are the basic feasible solutions.

There is one dimension per successor state and thus at most exponentially many basic feasible solutions.
This set can be computed in exponential time using standard theory of linear programs or simple geometric computations.
Due to lack of space we refer the reader to, e.g., \cite[Sec.~4]{DBLP:journals/tcs/HaddadM18} for further detail.

\subsection{Solving BMDP}\label{sec:evil}


When solving Rabin objectives on BMDPs, we observe one central problem:
The set of end components depends on the choice of intervals.
For example, consider the BMDP in \cref{fig:model_examples:bmdp}.
This BMDP comprises only one MEC containing all three states.
However, depending on the choice of probabilities, edges may be removed from the underlying graph, as for example in the right MDP of \cref{fig:model_examples:mdp}.
This MDP contains two MECs, namely $(\{q_1\}, \{b\})$ as well as $(\{q_2\}, \{d\})$.

In \cite{DC18}, a similar problem was encountered, as in IMC the set of \emph{bottom strongly connected components} can be modified by the choice of intervals.
The authors solved the problem implicitly by considering states where the Rabin objective is not satisfied as \enquote{leaky}, i.e. not part of any strongly connected component.

In contrast to that, our general solution relies on making the possible choices of the probabilities explicit.
We utilize the key observation that the non-determinism induced by intervals essentially corresponds to adding a player, who picks the probabilities for the intervals.
Thus, we reduce the BMDP to a stochastic game that can be solved by known model checking methods.
However, in the next section we introduce a more sophisticated approach inspired by the ideas of \cite{DC18}.

Intuitively, the reduction from BMDP to SG amounts to replacing every action in the BMDP with an additional state owned by the new player.
In this state, the new player can choose from the corresponding basic feasible solutions, allowing him to construct any consistent distribution over the successors.
An example of this construction is shown in \cref{fig:SG}, where the game constructed from the BMDP in \cref{fig:model_examples:bmdp} is depicted.
This new player can play against the system (yielding $\problower$) or cooperate (yielding $\probupper$).

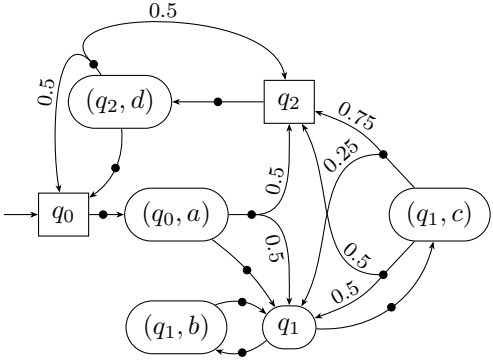
\begin{figure}
	\centering
	\begin{tikzpicture}[auto,initial text=]
	\node[maxstate,initial left] at (0,0) (q0) {$q_0$};
	\node[minstate] at (1.5,0) (q0a) {$(q_0, a)$};
	\node[actionnode] at (2.5,0) (q0a1) {};
	
	\node[minstate] at (3,-1.5) (q1) {$q_1$};
	\node[minstate] at (1.5,-1.5) (q1b) {$(q_1, b)$};
	\node[minstate] at (5,0) (q1c) {$(q_1, c)$};
	\node[actionnode] at (4.25,0.8) (q1c1) {};
	\node[actionnode] at (4.25,-0.8) (q1c2) {};
	
	\node[maxstate] at (3,1.5) (q2) {$q_2$};
	\node[minstate] at (0.75,1.5) (q2d) {$(q_2, d)$};
	\node[actionnode] at (0.4,2) (q2d1) {};
	
	\path[directedge]
		(q0) edge[pos=0.4] node[actionnode,anchor=center] {} (q0a)
		(q1) edge[bend left] node[actionnode,anchor=center] {} (q1b)
		(q1) edge[bend right=40] node[actionnode,anchor=center] {} (q1c)
		(q2) edge node[actionnode,anchor=center] {} (q2d)
		
		(q0a) edge[bend left=10] node[actionnode,anchor=center] {} (q1)
		(q1b) edge[bend left] node[actionnode,anchor=center] {} (q1)
		(q2d) edge[bend left] node[actionnode,anchor=center] {} (q0)
	;
	
	\path[actionedge]
		(q0a) edge (q0a1)
		(q1c) edge (q1c1)
		(q1c) edge (q1c2)
		(q2d) edge (q2d1)
	;
	
	\path[probedge]
		(q0a1) edge[in=90,out=0] node[probtext,sloped,anchor=north] {$0.5$} (q1)
		(q0a1) edge[in=-90,out=0] node[probtext,sloped,anchor=south] {$0.5$} (q2)
		
		(q1c1) edge[in=60,out=180,pos=0.2] node[probtext,sloped,anchor=south] {$0.25$} (q1)
		(q1c1) edge[in=-20,out=130] node[probtext,sloped,anchor=south] {$0.75$} (q2)
		
		(q1c2) edge[in=20,out=-130] node[probtext,sloped,anchor=south] {$0.5$} (q1)
		(q1c2) edge[in=-60,out=180,pos=0.2] node[probtext,sloped,anchor=south] {$0.5$} (q2)
		
		(q2d1) edge[in=100,out=135] node[probtext,sloped,anchor=south] {$0.5$} (q0)
		(q2d1) edge[in=100,out=135] node[probtext,sloped,anchor=south] {$0.5$} (q2)
	;
	\end{tikzpicture} %
	\caption{The stochastic game obtained from the BMDP in \cref{fig:model_examples:bmdp} by the construction in the proof of Theorem 2.
		States owned by the system are depicted by rectangles, while environment states have rounded corners.
		For readability, we omit all action labels.
		Furthermore, we omit action nodes for probability $1$ transitions.}
	\label{fig:SG}
\end{figure}
\begin{theorem}\label{stm:sg_equality}
	For every BMDP $\BMDP$, there exists an SG $\SG(\BMDP)$ such that 
	\begin{enumerate}[(i)]
		\item $\probupper(\BMDP) = \sup_{\strategySystem \in \StrategiesSystem} \sup_{\strategyEnvironment \in \StrategiesEnvironment} \Probability[\SG(\BMDP)]<\strategySystem,\strategyEnvironment>[\infinitepath \accepted \RabinAcc]$, and
		\item $\problower(\BMDP) = \sup_{\strategySystem \in \StrategiesSystem} \inf_{\strategyEnvironment \in \StrategiesEnvironment} \Probability[\SG(\BMDP)]<\strategySystem,\strategyEnvironment>[\infinitepath \accepted \RabinAcc]$
	\end{enumerate}
	
\end{theorem}
\begin{proof}
	Let $\BMDP = (\States, \initialstate, \Actions, \stateactions, \bmdpTransitionsLower, \bmdpTransitionsUpper, \labelling)$ be a BMDP.
	In contrast to IMC, where we found the basic feasible solutions for a state, in BMDP we have to consider state-action pairs.
	Recall that we assumed that each action belongs to a single state.
	Hence, each $a \in \Actions$ induces a set of basic feasible solutions, given by the constraints $\bmdpTransitionsLower(s, a)$ and $\bmdpTransitionsUpper(s, a)$, where $s$ is the unique state with $a \in \stateactions(s)$.
	We denote this set by $\BFS(a)$, and it can be computed as described in \cref{sec:bfs}.
	Recall that any $p \in \BFS(a)$ corresponds to a distribution over states.

	The SG $\SG(\BMDP) = (\States', \initialstate, \Actions', \stateactions', \sgTransitions, \RabinAcc, \ownerFunction)$ is constructed from $\BMDP$ as follows.
	\begin{itemize}
		\item $\States' = \States \union \{(s, a) \mid s \in \States \land a \in \stateactions(s)\}$,
		\item $\Actions' = \Actions \union \{(a, p) \mid a \in \Actions, p \in \BFS(a)\}$, and
		\item $\ownerFunction(s) = \systemPlayer$ if $s \in \States$ and $\environmentPlayer$ otherwise, i.e. all newly introduced states $(s, a)$ belong to the environment.
		\item For every old state $s \in \States$ we set
		\begin{itemize}
			\item $\stateactions'(s) = \stateactions(s)$ and
			\item $\sgTransitions(s, a, (s, a)) = 1$ (and 0 for all other states) for all actions $a \in \stateactions(s)$.
		\end{itemize}
		\item For every new state $(s, a) \in \States' \setminus \States$ we set
		\begin{itemize}
			\item $\stateactions'((s, a)) = \{(a, p) \mid p \in \BFS(a)\}$ and 
			\item $\sgTransitions((s, a), (a, p), s') = p(s')$ for all $p \in \BFS(a)$.
		\end{itemize}
	\end{itemize}
	For any consistent MDP $\MDP \in \BMDP$, there exists a policy for the other player $\strategyEnvironment$ inducing the transition function of $\MDP$, i.e. for all $\strategySystem \in \StrategiesSystem$ we have $\Probability[\MDP]<\strategySystem>[\infinitepath \accepted \RabinAcc] = \Probability[\SG(\BMDP)]<\strategySystem, \strategyEnvironment>[\infinitepath \accepted \RabinAcc]$, and vice versa.
	This can be proven completely analogously to the proof for IMC, see, e.g., \cite[Thm.~8]{CSH08}.
	Intuitively, randomizing over the BFS exactly yields the set of valid distributions.
	This immediately yields the desired equality.
\end{proof}

\begin{cor}
	Positional policies suffice to achieve optimal solutions in BMDP.
	Consequently, the uncertain and antagonistic semantics are equivalent for the optima.
\end{cor}
\begin{proof}
	Positional policies are sufficient for Rabin objectives in SG~\cite{CH12}, and thus by the reduction of \cref{stm:sg_equality} also for BMDP.
	Hence the best policy for choosing the intervals is positional both for maximization and minimization, and there is no benefit in switching.
\end{proof}
\begin{remark}
	This result relies on the fact that the objective $\RabinAcc$ is already part of the BMDP.
	See Appendix~\ref{sec:appendix:semantics} for further discussion.
\end{remark}
To solve $\omega$-regular objectives for SGs, we use the strategy improvement algorithm presented in \cite{CH06} and the reachability algorithm for MDPs from \cite{BaierBook}, which yields both the maximum and minimum probability, as well as the optimal controller.
Given a BMDP $\BMDP$, our procedure works as follows:
\begin{enumerate}
	\item Lower bound:
	\begin{enumerate}
		\item Use \cite[Alg.~2]{CH06} to solve $\SG(\BMDP)$ with the given Rabin objective $\RabinAcc$.
		This yields the optimal player 1 policy $\strategySystem$, which induces an MDP $\MDP_{\strategySystem}$.
		Note that the actions of the system, the original non-determinism of the BMDP, are fixed in $\MDP_{\strategySystem}$, and the remaining non-determinism belongs to the environment player choosing the intervals.\label{step:wcRabin}
		\item Compute the minimum reachability probability with the algorithm from \cite[Sec.~10.6.4]{BaierBook}, to obtain the the value $\problower$ and the worst-case environment policy $\strategyEnvironment$ for $\MDP_{\strategySystem}$, and thus the worst-case instantiation of the intervals.
	\end{enumerate}
	\item Upper bound:
	\begin{enumerate}
		\item Let $\MDP(\BMDP)$ be the MDP obtained by assigning all nodes in $\SG(\BMDP)$ to player $\systemPlayer$.\label{step:MM}
		\item Use standard model checking procedures to obtain the best policy $\strategy$ for $\MDP(\BMDP)$ and the value $\probupper$; the policy handles both the non-determinism of the system and the instantiation of the intervals.
	\end{enumerate}
\end{enumerate}

\begin{theorem}
	This procedure correctly computes
	\begin{enumerate}
		\item $\problower(\BMDP)$, the optimal policy $\strategySystem$ to achieve it and the worst-case MDP, and
		\item $\probupper(\BMDP)$, and $\strategy$, describing the best-case system controller as well as the best choice of intervals.
	\end{enumerate}
	It terminates in time exponential in the size of the $\SG(\BMDP)$, which in turn is exponential in the size of $\BMDP$.
\end{theorem}
\begin{proof}
	The correctness follows directly from \cref{stm:sg_equality} and the correctness of the the used algorithms \cite[Thm.~3]{CH06}, \cite[Sec.~10.6.4]{BaierBook}.
	The complexity is dominated by the computation of the game $\SG(\BMDP)$ and its solution process.
	The SG is exponential in the size of the BMDP (see \cref{sec:bfs}) and the solution algorithm takes time exponential in the size the game \cite[Thm.~3]{CH06}.
\end{proof}


\begin{remark}
	We can immediately apply the presented methods to \enquote{bounded-parameter stochastic games}, i.e. stochastic games with transition probability intervals.
	In particular, we do not need to introduce another player, as the states added by the construction in the proof of \cref{stm:sg_equality} can be assigned to one of the existing players -- player 1 in the uncertain and player 2 in the antagonistic setting.
	Thus the system remains a stochastic game.
\end{remark}

\subsection{Improving computation of upper bounds} \label{sec:design}

Note that for the computation of the upper bounds, we did not make use of the second player, but instead only introduced new states for the existing player, yielding an exponentially large MDP.
By adapting observations from \cite{DC18}, we can improve on this algorithm by directly analysing the MC and avoiding the explicit construction of the SG, yielding a polynomial time algorithm for computing the upper bound.
%

Intuitively, the improved procedure symbolically identifies for each Rabin pair in $(\RabinFin, \RabinInf) \in \RabinAcc$ the winning end components.
It does so by computing MECs while temporarily excluding states in $\RabinFin$.
After obtaining all states that are in a winning end-component for some pair, we compute the probability to reach any of these states.
The structure of this improved procedure is similar to \cite[Alg.~1]{DC18}, and relies on methods to compute the MEC decomposition on a BMDP from \cite{DBLP:journals/tcs/HaddadM18} and the reachability algorithm from \cite{PLSS13}.\footnote{In \cite{DBLP:journals/tcs/HaddadM18} and \cite{PLSS13}, the authors refer to BMDP as \enquote{IMDP}.}


Given a BMDP $\BMDP$, our procedure works as follows:
\begin{enumerate}
	\item For each Rabin pair $(\RabinFin_i, \RabinInf_i) \in \RabinAcc$: \label{alg:cooperative:rabin}
	\begin{enumerate}
		\item Construct a modified copy $\BMDP_i$ of $\BMDP$ where all $(s, q)$ with $q \in \RabinFin_i$ absorbing, i.e. all outgoing transitions are replaced with a self-loop. \label{alg:cooperative:absorbing}
		\item Compute the MEC decomposition of the resulting BMDP using \cite[Alg.~3]{DBLP:journals/tcs/HaddadM18}. \label{alg:cooperative:rabin_mecs}
		\item If a MEC $(T, A)$ has a non-empty intersection with $\RabinInf_i$, then it is winning and all states $T$ are added to the set of winning states $W$. \label{alg:cooperative:add_winning}
	\end{enumerate}
	\item Compute the maximal probability of reaching $W$, using the methods presented in, e.g., \cite{PLSS13}.
	\label{alg:cooperative:reach}
\end{enumerate}
\begin{theorem}
	This procedure correctly computes $\probupper(\BMDP)$ and terminates in polynomial time.
\end{theorem}
\begin{proof}
	By \cite[Prop.~6]{DBLP:journals/tcs/HaddadM18}, we get that the MEC computation of a BMDP $\BMDP$ through \cite[Alg.~3]{DBLP:journals/tcs/HaddadM18} is correct, i.e. the computed MECs correspond to the MECs of $\SG(\BMDP)$.
	The MECs identified as winning in Step~\ref{alg:cooperative:rabin_mecs} and \ref{alg:cooperative:add_winning} thus indeed are winning MECs in $\MDP(\BMDP)$, where $\MDP(\BMDP)$ is as in Step \ref{step:MM} of the procedure of \cref{sec:evil}.
	Consequently, we exactly identify the set of potentially winning states.
	Finally, the correctness of the algorithms used to compute the reachability in Step~\ref{alg:cooperative:reach} yields the overall correctness.
	
	For the complexity, observe that each step requires at most polynomial time (\cite[Prop.~6]{DBLP:journals/tcs/HaddadM18} for Step~\ref{alg:cooperative:rabin_mecs}, \cite[Thm.~4.1]{PLSS13} for Step~\ref{alg:cooperative:reach}) and there are only linearly many Rabin pairs in $\RabinAcc$.
\end{proof}

\begin{remark}
	Note that in the setting of \cite{DC18} a procedure for computing the upper bound is sufficient, as the lower bound can be computed as 1 minus the maximal probability to satisfy the negation of the specification.
	However, this idea is not applicable in the BMDP setting due to the alternation of $\sup$ and $\inf$ in the definition of $\problower$.
\end{remark}

\begin{remark}
	Our algorithm can easily be extended to use \emph{generalized Rabin transition acceptance}, allowing for efficient practical implementation.
	See \cite{DBLP:conf/cav/ChatterjeeGK13} for details.
\end{remark}

%
%
%
%

\section{Case Study}\label{sec:cs}

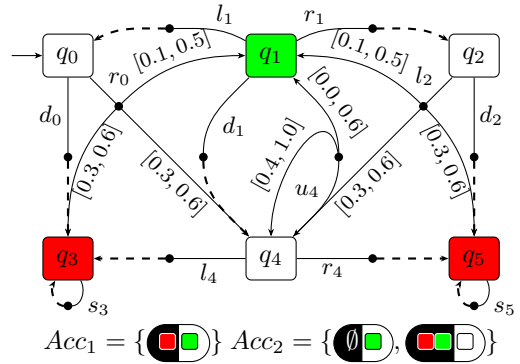
\begin{figure}[t]
	\centering
	\begin{tikzpicture}[auto,initial text=,scale=0.9]
		\node[state,initial left] at (0,0) (q0) {$q_0$};
		\node[state,fill=green] at (3,0) (q1) {$q_1$};
		\node[state] at (6,0) (q2) {$q_2$};
		\node[state,fill=red] at (0,-3) (q3) {$q_3$};
		\node[state] at (3,-3) (q4) {$q_4$};
		\node[state,fill=red] at (6,-3) (q5) {$q_5$};

		\node[actionnode] at (0,-1.5) (q0d) {};
		\node[actionnode] at (0.75,-0.75) (q0r) {};

		\node[actionnode] at (2,-1.5) (q1d) {};
		\node[actionnode] at (1.5,0.4) (q1l) {};
		\node[actionnode] at (4.5,0.4) (q1r) {};

		\node[actionnode] at (6,-1.5) (q2d) {};
		\node[actionnode] at (5.25,-0.75) (q2l) {};

		\node[actionnode] at (0,-3.7) (q3s) {};

		\node[actionnode] at (1.5,-3) (q4l) {};
		\node[actionnode] at (4.5,-3) (q4r) {};
		\node[actionnode] at (4,-1.5) (q4u) {};

		\node[actionnode] at (6,-3.7) (q5s) {};

		\path[actionedge]
			(q0) edge node[actiontext] {$r_0$} (q0r)
			(q0) edge node[actiontext,swap] {$d_0$} (q0d)

			(q1) edge[in=0,out=150] node[actiontext,swap] {$l_1$} (q1l)
			(q1) edge[in=90,out=-135] node[actiontext] {$d_1$} (q1d)
			(q1) edge[out=30,in=180] node[actiontext] {$r_1$} (q1r)

			(q2) edge node[actiontext,swap] {$l_2$} (q2l)
			(q2) edge node[actiontext] {$d_2$} (q2d)

			(q3) edge[out=-60,in=0] node[actiontext] {$s_3$} (q3s)

			(q4) edge node[actiontext] {$l_4$} (q4l)
			(q4) edge[in=-90,out=45] node[actiontext] {$u_4$} (q4u)
			(q4) edge node[actiontext,swap] {$r_4$} (q4r)

			(q5) edge[out=-60,in=0] node[actiontext] {$s_5$} (q5s)
		;

		\path[probedge]
			(q0r) edge[sloped,in=180,out=45] node[anchor=south,probtext] {$[0.1,0.5]$} (q1)
			(q0r) edge[sloped] node[anchor=north,probtext] {$[0.3,0.6]$} (q4)
			(q0r) edge[sloped,in=90,out=225,pos=0.4] node[anchor=north,probtext] {$[0.3,0.6]$} (q3)

			(q0d) edge[thick,dashed] node[probtext] {} (q3)

			(q1l) edge[out=180,in=30,thick,dashed] node[probtext] {} (q0)

			(q1r) edge[out=0,in=150,thick,dashed] node[probtext] {} (q2)

			(q1d) edge[out=-90,in=135,thick,dashed] node[probtext] {} (q4)

			(q2l) edge[sloped,in=0,out=135] node[anchor=south,probtext] {$[0.1,0.5]$} (q1)
			(q2l) edge[sloped] node[anchor=north,probtext] {$[0.3,0.6]$} (q4)
			(q2l) edge[sloped,in=90,out=315,pos=0.4] node[anchor=north,probtext] {$[0.3,0.6]$} (q5)

			(q2d) edge[thick,dashed] node[probtext] {} (q5)

			(q3s) edge[out=-180,in=-120,thick,dashed] node[probtext] {} (q3)

			(q4l) edge[thick,dashed] node[probtext] {} (q3)

			(q4u) edge[out=100,in=90,sloped,pos=0.6,looseness=1.5] node[anchor=south,probtext] {$[0.4,1.0]$} (q4)
			(q4u) edge[out=90,in=-45,sloped] node[anchor=south,probtext] {$[0.0,0.6]$} (q1)

			(q4r) edge[thick,dashed] node[probtext] {} (q5)

			(q5s) edge[out=-180,in=-120,thick,dashed] node[probtext] {} (q5)
		;

		\node[anchor=west] at (-0.5,-4.25) {$\RabinAcc_1 = \{ \rabinpair{ \colorchar{red} }{ \colorchar{green} } \}$};
		\node[anchor=east] at (6.5,-4.25) {$\RabinAcc_2 = \{ \rabinpair{ \emptyset }{ \colorchar{green} }, \rabinpair{ \colorchar{red}\colorchar{green} }{ \colorchar{white} } \}$};
	\end{tikzpicture}
%
%
%
%
	\caption[]{Graphic representation of our extension of the case study from \cite{DC18}.
		A robot navigates through a grid of six states according to the BMDP, using the actions \textbf{u}p, \textbf{d}own, \textbf{l}eft, and \textbf{r}ight, where enabled.
		For readability, all $[1.0, 1.0]$ constraints are omitted and the corresponding edges are drawn dashed.
		In our analysis, we consider two different acceptance conditions $\RabinAcc_1$ and $\RabinAcc_2$.
	} \label{fig:cs}
\end{figure}

We extend the case study from \cite{DC18}, where an agent moves on a coloured grid of six states.
We added several actions in each state, modelling a robot which navigates the grid as depicted in \cref{fig:cs}.
It can choose between going down, up, left or right; it cannot leave the grid, e.g., in $q_0$ only down and right are available.
The robot is pulled towards the red states $q_3$ and $q_5$, e.g. when moving right from $q_0$ there is some chance to be pulled down onto $q_4$ or even $q_3$.
The strength of the pulling force and hence also the probability distribution over the successor states is unknown and hence modelled by uncertainty intervals.

As a first example, we calculate the probability to reach state $q_2$ starting from $q_0$.
From $q_1$ we can surely reach $q_2$ and from $q_3$ and $q_5$ there is no possibility of reaching $q_2$.
From $q_4$, it depends on the instantiation of the BMDP.
The probability to remain in $q_4$ may equal $1$, preventing the robot from reaching $q_2$ when it is in $q_4$.
On the other hand, we might have $\mdpTransitions(q_4, u_4, q_1) = 0.6$ and $q_1$ can almost surely be reached from $q_4$.
Consequently, $q_2$ can be reached from $q_4$ with probability $1$.
Finally, the best strategy in $q_0$ is to go right, as otherwise the robot is immediately stuck in $q_3$.
Together, this gives us a lower bound of $0.1$ and an upper bound of $0.7$.

\begin{table}[t]
	\caption{Results for the BMDP in \cref{fig:cs}} \label{fig:results}
	\centering
	\begin{tabular}{rrr}
		& $\problower$ & $\probupper$ \\
		\toprule
		$\RabinAcc_1$ & $0.0$ & $0.7$ \\
		$\RabinAcc_2$ & $0.4$ & $0.7$ \\
		\bottomrule
	\end{tabular}
\end{table}

Now we consider the properties from \cite{DC18}, adjusted to our setting as depicted in \cref{fig:cs}.
The resulting probabilities are given in \cref{fig:results} and explained below.

$\RabinAcc_1$ corresponds to \enquote{The agent visits a green state infinitely often while visiting red states finitely often}.
In the antagonistic interpretation, the probability of satisfying this property is zero, since there is no MEC containing $q_1$.
Intuitively, actions $r_0$, $d_1$ and $l_2$ all have a positive probability of moving to the bottom row, where we may be forced to remain forever.
For the upper bound, observe that by setting $\mdpTransitions(q_4, u_4, q_1) = 0.6$ we obtain the winning MEC $(\{q_1,q_4\}, \{d_1,u_4\})$, which can be reached with probability $0.7$.

$\RabinAcc_2$ corresponds to \enquote{The agent visits a red state infinitely often only if it visits a green state infinitely often}.
Observe that in this case both $q_4$ and $q_1$ are winning in any consistent MDP, as playing actions $d_1$ and $u_4$ always leads to a winning path.
In contrast to $\RabinAcc_1$, remaining in $q_4$ is winning by the second pair of $\RabinAcc_2$, since only white states are visited.
We thus get a lower and upper bound of $0.4$ and $0.7$, respectively, by computing the probability to reach $q_1$ or $q_4$.

\section{Conclusion} \label{sec:conclusion}

We have presented a solution to the open problem of bounding the probabilities to satisfy an $\omega$-regular property on a bounded-parameter Markov systems.
A~different perspective on previous approaches enabled us to solve the problem by analysis of $\omega$-regular stochastic games.
Future work includes applications of our approach to more general settings such as MDPIP, as well as a practical implementation.
For the latter, we plan to apply approaches based on learning and real-time dynamic programming, see e.g.~\cite{FDB14}.

\bibliographystyle{plain}
\bibliography{main}

\section{Appendix -- The product construction} \label{sec:appendix}

In this section, we briefly explain how linear objectives usually are specified and how they are model checked using the product construction.
First, we define the concept of \emph{$\omega$-regular languages} and \emph{automata}.
Fix a finite \emph{alphabet} $\Alphabet$.
Elements of $\Alphabet^\omega$ are called \emph{words} and sets of words $\Language \subseteq \Alphabet^\omega$ are called \emph{languages}.
Such a language is $\omega$-regular if it can be \emph{recognized} by an automaton.

\subsection{Automata \& regular languages}

\begin{definition}
	A \emph{deterministic Rabin automaton (DRA)} is a tuple $\automaton = (\automatonStates, \automatonTransitions, \automatonInitialState, \RabinAcc)$, where
		$\automatonStates$ is a finite set of \emph{states},
		$\automatonTransitions : \automatonStates \times \Alphabet \to \automatonStates$ is a \emph{transition function}\footnote{Recall that the alphabet $\Alphabet$ is already fixed.},
		$\automatonInitialState \in \automatonStates$ is an \emph{initial state}, and
		$\RabinAcc \subseteq 2^\automatonStates \times 2^\automatonStates$ is the \emph{Rabin condition}.
	An element $(\RabinFin_i, \RabinInf_i) \in \RabinAcc$ is called \emph{Rabin pair}.
	We assume w.l.o.g. that $\RabinFin_i \intersection \RabinInf_i = \emptyset$.
\end{definition}
Let $\automaton$ be a DRA and $w \in \Alphabet^\omega$ a word.
A word $w$ induces a \emph{run}, i.e. a sequence of states $\automaton(w) = \automatonInitialState q_1 q_2 \cdots \in \automatonStates^\omega$, where $q_{i+1} = \automatonTransitions(q_i, w_i)$.
As with MDP, let $\Inf(w)$ denote the set of states occurring \emph{infinitely often} on the run $\automaton(w)$.
A word is accepted by the automaton, denoted $w \accepted \automaton$, if there exists a Rabin pair $(\RabinFin_i, \RabinInf_i) \in \RabinAcc$ with $\RabinFin_i \intersection \Inf(w) = \emptyset$ and $\RabinInf_i \intersection \Inf(w) \neq \emptyset$.
Such a Rabin pair is called \emph{accepting for $w$}.

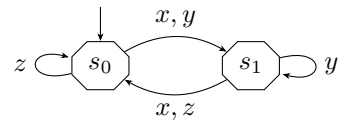
\begin{figure}
	\centering
	\begin{tikzpicture}[auto,initial text=]
		\node[automatonstate,initial above] at (0,0) (s0) {$s_0$};
		\node[automatonstate] at (2,0) (s1) {$s_1$};

		\path[->]
			(s0) edge[loop left] node {$z$} (s0)
			(s0) edge[bend left] node {$x, y$} (s1)
			(s1) edge[loop right] node {$y$} (s1)
			(s1) edge[bend left] node {$x, z$} (s0)
		;
	\end{tikzpicture}
	\caption[]{Example DRA with acceptance $\RabinAcc = \{\rabinpair{s_0}{s_1}, \rabinpair{s_1}{s_0}\}$.} \label{fig:model_examples:dra}
\end{figure}

A language $\Language \subseteq \Alphabet^\omega$ is called \emph{$\omega$-regular} if and only if there exists a DRA $\automaton$ recognizing $\Language$, i.e. some word $w \in \Alphabet^\omega$ is accepted by the automaton if and only if it is in $\Language$.
See \cref{fig:model_examples:dra} for an example of a DRA recognizing the language \enquote{eventually only $y$ or eventually only $z$}.

\begin{remark}
	A wide variety of specifications are $\omega$-regular.
	For example, reachability and liveness constraints can easily be translated to an automaton.
	Moreover, the whole of linear temporal logic is expressible through Rabin automata and efficient translations from LTL to Rabin automata exist \cite{DBLP:conf/lics/EsparzaKS18}.
\end{remark}

\subsection{Labelled MDPs \& product}

For the product construction, we modify the definition of MDPs by replacing the acceptance by a labelling function $\labelling : \States \to \Alphabet$, assigning to each state of the MDP a letter.
We are given such a labelled MDP and a Rabin automaton.
We construct the product by tracking both the evolution of the MDP and the automaton, where the automaton progresses based on the letter assigned to the current state.

\begin{definition} \label{def:product}
	Let $\MDP = (\States, \initialstate, \Actions, \stateactions, \mdpTransitions, \labelling)$ be a labelled MDP and $\automaton = (\automatonStates, \automatonTransitions, \automatonInitialState, \RabinAcc)$ a Rabin automaton.
	The \emph{product} $\MDP \otimes \automaton = (\States \times \automatonStates, (\initialstate, \automatonInitialState), \Actions, \stateactions', \mdpTransitions', \RabinAcc')$ is an MDP where
		$\stateactions'((s, q)) := \stateactions(s)$,
		$\mdpTransitions((s, q), a, (s', q')) := \mdpTransitions(s, a, s')$ if $q' = \automatonTransitions(q, \labelling(s))$ and $0$ otherwise, and
		$\RabinAcc' = \{(\RabinFin_i \times \States, \RabinInf_i \times \States) \mid (\RabinFin_i, \RabinInf_i) \in \RabinAcc\}$.
\end{definition}
We analogously define this product construction for BMDP.
Observe that the product is now of the form as we defined it in \cref{def:product}.
Applying our methods to this product yields a solution for the original system.

\subsection{Caveat} \label{sec:appendix:semantics}

\begin{figure}
	\centering
	\subfloat[][Example DRA with acceptance $\RabinAcc = \{\rabinpair{q_2}{q_0, q_1}\}$]{ \label{fig:product:dra}
		\begin{tikzpicture}[auto,initial text=]
			\node[automatonstate,initial left] at (0,0) (q0) {$q_0$};
			\node[automatonstate] at (2,0) (q1) {$q_1$};
			\node[automatonstate] at (4,0) (q2) {$q_2$};
			\node[automatonstate] at (6,0) (q3) {$q_3$};
			\node[automatonstate] at (3,-1.5) (e) {$q_4$};
	
			\path[->]
				(q0) edge node {$x$} (q1)
				(q1) edge node {$y$} (q2)
				(q2) edge node {$x$} (q3)
				(q3) edge[bend right,swap] node {$z$} (q0)

				(q0) edge[sloped,anchor=north] node {$y, z$} (e)
				(q1) edge[sloped,anchor=south] node {$x, z$} (e)
				(q2) edge[sloped,anchor=south] node {$y, z$} (e)
				(q3) edge[sloped,anchor=north] node {$x, y$} (e)

				(e) edge[loop left] node {$x, y, z$} (e)
			;
		\end{tikzpicture}
	} \\
	\subfloat[Example IMC.]{  \label{fig:product:imc}
		\begin{tikzpicture}[auto,initial text=]
			\node[state,initial above] at (0,0) (s0) {$s_0, x$};
			\node[state] at (2,0) (sr) {$s_2, z$};
			\node[state] at (-2,0) (sl) {$s_1, y$};

			\path[probedge]
				(s0) edge[bend left] node[probtext] {$[0.0, 1.0]$} (sr)
				(s0) edge[bend left] node[probtext] {$[0.0, 1.0]$} (sl)
				(sl) edge[bend left] node[probtext] {$[1.0, 1.0]$} (s0)
				(sr) edge[bend left] node[probtext] {$[1.0, 1.0]$} (s0)
			;
		\end{tikzpicture}
	}
	\caption{Uncertain/adversarial interval interpretations are not equal.} \label{fig:product}
\end{figure}
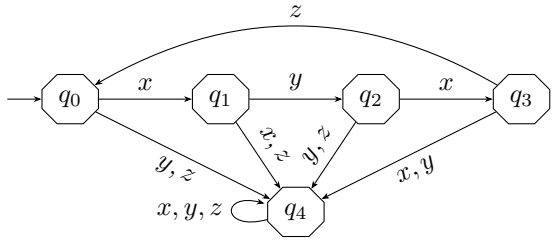
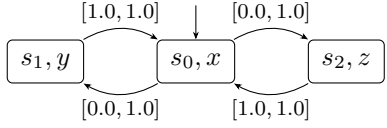

Since this construction modifies the state space, it is not obvious how optimal policies on the product relate to policies on the original MDP.
Indeed, while a memoryless policy may be optimal on the product, it might be the case that finite memory is needed to behave optimally in the given system.
Moreover, it is the case that for $\omega$-regular objectives, the optimal values for the uncertainty and the antagonistic interpretation are not equal already for IMCs.

Fix the alphabet $\Alphabet = \{x, y, z\}$ and consider the language $\Language = \{(xyxz)^\omega\}$, i.e. a language containing a single word $w$ which repeats the string $xyxz$ indefinitely.
This language is regular and is recognized by the automaton shown in \cref{fig:product:dra}.
Consider now the IMC in \cref{fig:product:imc}.
Clearly, the probability of satisfying the property is zero under the uncertainty interpretation -- any MC consistent with the IMC eventually violates the structure of the property with probability $1$.
On the other hand, interpreted as an IMDP, the transitions can be chosen such that the required word is always produced.

\end{document}